\documentclass[a4paper,11pt]{article}

\usepackage{complexity} 
\usepackage{enumerate}
\usepackage{amsmath, amssymb, amsthm,complexity}
\usepackage{euler}
\usepackage{authblk}
\usepackage{todonotes}

\def\margin{2.9cm}
\usepackage[left=\margin,right=\margin,top=\margin,bottom=\margin]{geometry}

\usepackage{thmtools}
\usepackage{thm-restate}
\usepackage{hyperref}
\usepackage{cleveref}

\usepackage{authblk}

\theoremstyle{plain}
\newtheorem{theorem}{Theorem}

\newtheorem{lemma}[theorem]{Lemma}
\newtheorem{corollary}[theorem]{Corollary}

\newtheorem{preprocess rule}{Preprocessing Rule}[section]
\newtheorem{reduction rule}{Reduction Rule}[section]
\newtheorem{branching rule}{Branching Rule}[section]
\usepackage{complexity}

\title{Revisiting Connected Vertex Cover: FPT Algorithms and Lossy Kernels}

\date{}
\author{R. Krithika}
\author{Diptapriyo Majumdar}
\author{Venkatesh Raman}
\affil{The Institute of Mathematical Sciences, HBNI, Chennai, India.\\
  \texttt{\{rkrithika|diptapriyom|vraman\}@imsc.res.in}}

\usepackage{todonotes}
\usepackage{microtype}
\usepackage{multirow}
\usepackage{amsmath,amssymb}
\newtheorem{observation}{Observation}
\usepackage{comment}
\usepackage{enumerate}

\usepackage{xspace}

\usepackage{listings}
\usepackage{color}
\usepackage{cite}
\usepackage{complexity}

\usepackage{graphicx}
\RequirePackage{fancyhdr}
 \usepackage{xcolor}
\usepackage{boxedminipage}

\newcommand{\CC}{{\mathcal C}}

\newcommand{\OO}{{\mathcal O}}

\newcommand{\nka}{\NP\ $\subseteq$ \coNP/poly}
\newcommand{\cvc}{\textsc{Connected Vertex Cover}}

\DeclareMathOperator{\comp}{\#comp}

\begin{document}

\maketitle

\begin{abstract}
The \cvc\ problem asks for a vertex cover in a graph that induces a connected subgraph. The problem is known to be fixed-parameter tractable (\FPT), and is unlikely to have a polynomial sized kernel (under complexity theoretic assumptions) when parameterized by the solution size. In a recent paper, Lokshtanov et al. [STOC 2017], have shown an $\alpha$-approximate kernel for the problem for every $\alpha > 1$, in the framework of approximate or lossy kernelization. We exhibit lossy kernels and \FPT\ algorithms for \cvc\ for parameters that are more natural and functions of the input, and in some cases, smaller than the solution size.

Our first result is a lossy kernel for \cvc\ parameterized by the size $k$ of a split deletion set. A split graph is a graph whose vertex set can be partitioned into a clique and an independent set and a split deletion set is a set of vertices whose deletion results in a split graph. Let $n$ denote the number of vertices in the input graph. We show that 
\begin{itemize}
\item \cvc\ parameterized by the size $k$ of a split deletion set admits an $\alpha$-approximate kernel with $\mathcal{O}(k+(2k+\lceil \frac{2\alpha-1}{\alpha-1} \rceil)^{\lceil \frac{2\alpha-1}{\alpha-1} \rceil})$ vertices and a $3^k n^{\OO(1)}$ time algorithm. 
\item For the special case when the split deletion set is a clique deletion set, the algorithm runs in $2^k n^{\OO(1)}$ time and the lossy kernel has $\OO(k + \lceil \frac{2\alpha-1}{\alpha-1} \rceil)$ vertices.
\end{itemize}
To the best of our knowledge, this (approximate) kernel is one of the few lossy kernels for problems parameterized by a structural parameter (that is not solution size). We extend this lossy kernelization to \cvc\ parameterized by an incomparable parameter, and that is the size $k$ of a clique cover. A clique cover of a graph is a partition of its vertex set such that each part induces a clique. We show that
\begin{itemize}
\item \cvc\ parameterized by the size $k$ of a clique cover is \W[1]-hard but admits an $\alpha$-approximate kernel with $\mathcal{O}(k \lceil \frac{2\alpha-1}{\alpha-1} \rceil)$ vertices for every $\alpha > 1$. This is one of the few problems that are not \FPT\ but admit a lossy kernel.
\end{itemize}
Then, we consider the size of a cluster deletion set as parameter. A cluster graph is a graph in which every component is a complete graph and a cluster deletion set is a set of vertices whose deletion results in a cluster graph. We show that
\begin{itemize}
\item \cvc\ parameterized by the size $k$ of a cluster deletion set is \FPT\ via an $4^k n^{\OO(1)}$ time algorithm and admits an $\alpha$-approximate kernel with $\OO(k^2+\lceil \frac{2\alpha-1}{\alpha-1} \rceil \cdot \frac{k}{\alpha-1}+ \lceil \frac{\alpha}{\alpha-1} \rceil \cdot k^{\lceil \frac{\alpha}{\alpha-1} \rceil})$ vertices for every $\alpha > 1$.
\item For the special case when the cluster deletion set is a degree-$1$ modulator (a subset of vertices whose deletion results in a graph with maximum degree $1$), the \FPT\ algorithm runs in $3^k n^{\OO(1)}$ time and the lossy kernel has $\OO(k^2+\frac{k}{\alpha-1}+ \lceil \frac{\alpha}{\alpha-1} \rceil \cdot k^{\lceil \frac{\alpha}{\alpha-1} \rceil})$ vertices.
\end{itemize}
Finally, we consider \cvc\ parameterized by the size $k$ of a chordal deletion set. A chordal graph is a graph in which every cycle of length at least $4$ has a chord -- an edge joining non-adjacent vertices of the cycle. A chordal deletion set of a graph is a subset of vertices whose deletion results in a chordal graph. We show that \cvc\ parameterized by $k$ is \FPT\ using the known algorithm for \cvc\ on bounded treewidth graphs. 
\end{abstract}

\section{Introduction, Motivation and Our Results}
\label{sec:intro}
A {\em vertex cover} in a graph is a set of vertices that has at least one endpoint from every edge of the graph. In the \textsc{Vertex Cover} problem, given a graph $G$ and an integer $\ell$, the task is to determine if $G$ has a vertex cover of size at most $\ell$. Undoubtedly, \textsc{Vertex Cover} is one of the most well studied problems in parameterized complexity. In this framework, each problem instance is associated with a non-negative integer called {\em parameter}. The pair consisting of a decision problem and a parameterization is called a {\em parameterized problem}. A common parameter is a bound on the size of an optimum solution for the problem instance. A problem is said to be {\em fixed-parameter tractable} (\FPT) with respect to parameter $k$ if it can be solved in $f(k) n^{\mathcal{O}(1)}$ time for some computable function $f$, where $n$ is the input size. Such an algorithm is called a {\em parameterized algorithm} or {\em \FPT\ algorithm}. For convenience, the running time $f(k)n^{\mathcal{O}(1)}$ where $f$ grows super-polynomially with $k$ is denoted as $\mathcal{O}^*(f(k))$. A {\em kernelization algorithm} is a polynomial-time algorithm that transforms an arbitrary instance of the problem to an equivalent instance of the same problem whose size is bounded by some computable function $g$ of the parameter of the original instance. The resulting instance is called a {\em kernel} and if $g$ is a polynomial function, then it is called a polynomial kernel and we say that the problem admits a polynomial kernel. In order to classify parameterized problems as being \FPT\ or not, the \W-hierarchy: \FPT\ $\subseteq$ \W$[1] \subseteq$ \W$[2] \subseteq \cdots \subseteq$ \XP\ is defined. It is believed that the subset relations in this sequence are all strict and a parameterized problem that is hard for some complexity class above \FPT\ in this hierarchy is unlikely to be \FPT. A parameterized problem is said to be in \XP\ if it has an algorithm with runtime $f(k)n^{g(k)}$ time for some computable functions $f$ and $g$. A problem is said to be para-\NP-hard if it is not in \XP\ unless \P=\NP. The complexity classes \FPT\ and para-\NP\ can be viewed as the parameterized analogues of \P\ and \NP.

\textsc{Vertex Cover} is \FPT\ via an easy $\OO^*(2^{\ell})$ algorithm and after a long race, the current fastest algorithm runs in $\OO^*(1.2738^{\ell})$ time~\cite{CKJ01}. \textsc{Vertex Cover} also has a kernel with $2\ell - \OO(\log \ell)$ vertices and $\OO(\ell^2)$ edges~\cite{LNRRS14}. In the \cvc\ problem, we seek a {\em connected vertex cover}, i.e., a vertex cover that induces a connected subgraph. It is easy to observe that \textsc{Vertex Cover} reduces to \cvc\ implying that the latter is at least as hard as the former in general graphs. In fact, \cvc\ is \NP-hard even on bipartite graphs~\cite{EGM10} where \textsc{Vertex Cover} is solvable in polynomial time (Theorem $2.1.1$~\cite{Diestel12book}). 
However, \cvc\ is polynomial-time solvable on chordal graphs and sub-cubic graphs~\cite{EGM10,UKG88}
and has an $\OO^*(2^{\OO(t)})$ algorithm when the input graph has treewidth upper bounded by $t$~\cite{BCKN15}.
The best known parameterized algorithm for \cvc\ takes $\OO^*(2^{\ell})$ time where $\ell$ is the size of the connected vertex cover that we are looking for~\cite{Cygan12swat}. 
Further, the problem does not admit a polynomial kernel unless the polynomial hierarchy collapses (specifically, unless \nka)~\cite{DLS14}.

As the goal in parameterized algorithms is to eventually solve the given instance of a problem, the application of a (classical) kernelization algorithm is typically followed by an exact or approximation algorithm that finds a solution for the reduced instance. However, the current definition of kernels provide no insight into how this solution relates to a solution for the original instance and the basic framework is not amenable to be a precursor to approximation algorithms or heuristics. Recently, Lokshtanov et al.~\cite{LPRS16} proposed a new framework, referred to as {\em lossy kernelization}, that is a bit less stringent than the notion of polynomial kernels and combines well with approximation algorithms and heuristics. The key new definition in this framework is that of an {\em $\alpha$-approximate kernelization}. The precise definitions are deferred to Section \ref{sec:prelims}. Informally, an $\alpha$-approximate polynomial kernelization (lossy kernelization) is a polynomial time preprocessing algorithm that takes as input an instance of a parameterized problem and outputs another instance of the same problem whose size is bounded by a polynomial function of the parameter of the original instance. Additionally, for every $c \geq 1$, a $c$-approximate solution for the reduced instance can be turned into a $(\alpha c)$-approximate solution for the original instance in polynomial time. The authors of \cite{LPRS16} exhibit lossy polynomial kernels for several problems that do not admit classical polynomial kernels including \cvc\ parameterized by the solution size. In this paper, we extend this lossy kernel for \cvc\ to parameters that are more natural and functions of the input, and in some cases, smaller than the solution size.

In the initial work on parameterized complexity, the parameter was almost always the solution size (with treewidth being a notable exception).  A recent trend is to study the complexity of the given problem with respect to structural parameters that are more likely to be small in practice. Also, once a problem is shown to be \FPT\ or to have a polynomial sized kernel by a parameterization, it is natural to ask whether the problem is \FPT\ (and admits a polynomial kernel) when parameterized by a smaller parameter. Similarly, once a problem is shown to be \W-hard by a parameterization, it is natural to ask whether the problem is \FPT\ when parameterized by a larger parameter. Structural parameterizations of {\sc Vertex Cover} \cite{FJR13,BJ13}, {\sc Feedback Vertex Set}~\cite{JRV14} and {\sc Graph Coloring}~\cite{JK13} have been explored extensively. We refer to \cite{FJR13} and \cite{jansen-thesis} for a detailed introduction to the whole program. A parameter that has gained significant attention recently is the size of a {\em modulator} to a family of graphs. Let $\mathcal{F}$ denote a hereditary graph class (which is closed under induced subgraphs) on which \textsc{Vertex Cover} is polynomial-time solvable. Suppose $S$ is a set of vertices (called a $\mathcal{F}$-modulator) such that $G-S \in \mathcal{F}$. Then, \textsc{Vertex Cover} is \FPT\ when parameterized by the size of $|S|$. Following is an easy $\OO^*(2^{|S|})$ algorithm: guess the intersection of the required solution with the modulator and solve the problem on the remaining graph in polynomial time. In contrast, a similar generic result does not seem easy for \cvc\ and a study of such structural parameterizations for \cvc\ is another goal of this paper.

\smallskip
\noindent 
{\bf Our Results. }The starting point of our study is to extend the lossy kernelization known for \cvc\ parameterized by solution size to \cvc\ parameterized by a smaller parameter. A {\em split graph} is a graph whose vertex set can be partitioned into a clique and an independent set. Split graphs can be recognized in polynomial time and such a partition can be obtained in the process~\cite{golumbic}. A {\em split deletion set} is a set of vertices whose deletion results in a split graph. \cvc\ has an easy polynomial time algorithm in split graphs, and it is easy to verify that the size of minimum connected vertex cover is at least the size of the minimum vertex cover which is at least the size of the minimum split deletion set. We show that
\begin{itemize}
\item
\cvc\ parameterized by the size $k$ of a split deletion set is \FPT\ via an $\OO^*(3^k)$ algorithm and that the $\alpha$-approximate kernel for \cvc\ parameterized by solution size can be extended in a non-trivial way to an $\alpha$-approximate kernel with $\mathcal{O}(k+(2k+\lceil \frac{2\alpha-1}{\alpha-1} \rceil)^{\lceil \frac{2\alpha-1}{\alpha-1} \rceil})$ vertices for every $\alpha > 1$. This adds to the small list of problems for which lossy kernels with respect to structural parameterizations are known. 
\item
For the special case when the split deletion set is a {\em clique deletion set} (a set of vertices whose deletion results in a complete graph), the algorithm runs in $\OO^*(2^k)$ time and the lossy kernel has $\OO(k + \lceil \frac{2\alpha-1}{\alpha-1} \rceil)$ vertices (linear size). 
\end{itemize}
Then, we consider \cvc\ parameterized by the size of a {\em clique cover}. A {\em clique cover} of a graph is a partition of its vertex set such that each part induces a clique. We show that
\begin{itemize}
\item \cvc\ parameterized by the size $k$ of a clique cover is \W[1]-hard but admits an $\alpha$-approximate kernel with $\mathcal{O}(k \lceil \frac{2\alpha-1}{\alpha-1} \rceil)$ vertices for every $\alpha > 1$. This is one of the few problems that are not \FPT\ but admit a lossy kernel.
\end{itemize}
Then, we consider a parameter related (structurally) to clique cover size, namely, the size of a {\em cluster deletion set} and show that \cvc\ is \FPT\ with respect to this parameter. A {\em cluster graph} is a graph in which every component is a complete graph and a cluster deletion set is a set of vertices whose deletion results in a cluster graph. We show that
\begin{itemize}
\item \cvc\ parameterized by the size $k$ of a cluster deletion set is \FPT\ via an $\OO^*(4^k)$ algorithm and admits an $\alpha$-approximate kernel with $\OO(k^2+\lceil \frac{2\alpha-1}{\alpha-1} \rceil \cdot \frac{k}{\alpha-1}+ \lceil \frac{\alpha}{\alpha-1} \rceil \cdot k^{\lceil \frac{\alpha}{\alpha-1} \rceil})$ vertices for every $\alpha > 1$.
\item For the special case when the cluster deletion set is a degree-$1$ modulator (a subset of vertices whose deletion results in a graph with maximum degree $1$), the algorithm runs in $\OO^*(3^k)$ time and the lossy kernel has $\OO(k^2+\frac{k}{\alpha-1}+ \lceil \frac{\alpha}{\alpha-1} \rceil \cdot k^{\lceil \frac{\alpha}{\alpha-1} \rceil})$ vertices.
\end{itemize}
Finally, we consider \cvc\ parameterized by the size $k$ of a {\em chordal deletion set}. A {\em chordal graph} is a graph in which every cycle of length at least $4$ has a chord -- an edge joining non-adjacent vertices of the cycle. A chordal deletion set of a graph is a subset of vertices whose deletion results in a chordal graph. We show that \cvc\ parameterized by $k$ is \FPT\ using the known algorithm for \cvc\ on bounded treewidth graphs. A summary of these results is presented in Figure \ref{fig:parameter-hierarchy}.  
\begin{figure}[ht]
\centering
	\includegraphics[scale=0.35]{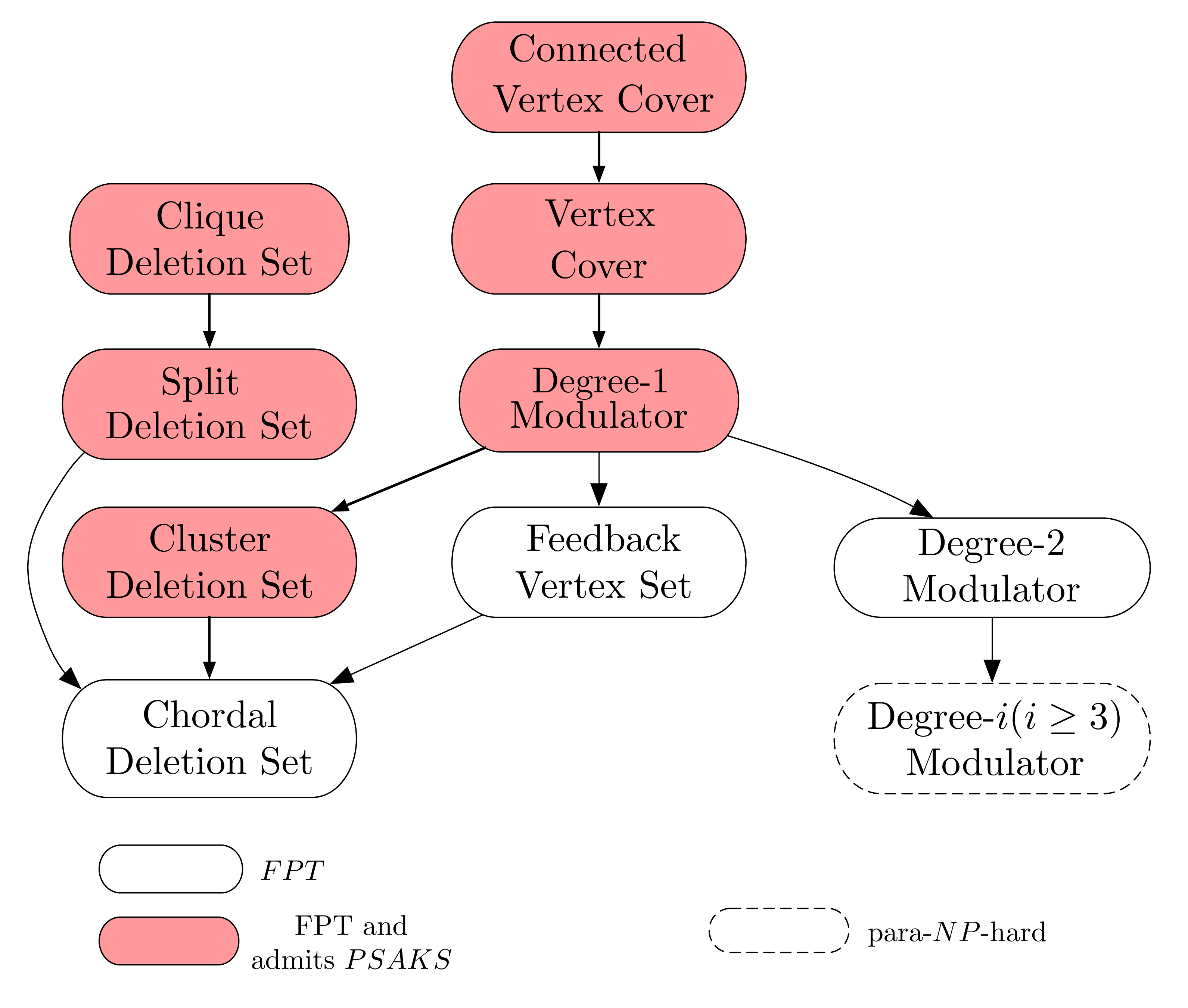}
	\caption{Ecology of Parameters for \cvc. The parameter values are the minimum possible for a given graph. An arrow from parameter $x$ to parameter $y$ means $y \leq x$.}
	\label{fig:parameter-hierarchy}
\end{figure}

\smallskip
\noindent 
{\bf Techniques. }Our FPT algorithms involve a reduction to the \textsc{Steiner Tree} problem on bipartite graphs. Given a graph $G$, an integer $p$ and a set $T$ (called {\em terminals}) of vertices of $G$, the \textsc{Steiner Tree} problem is the task of determining whether $G$ contains a tree (called {\em Steiner tree}) on at most $p$ vertices that contains $T$. A {\em bipartite graph} is a graph whose vertex set can be partitioned into 2 independent sets. Such a partition is called a {\em bipartition}. We use the following result known for \textsc{Steiner Tree} on bipartite graphs.

\begin{lemma}{\em \cite{Cygan12swat,steiner-best}}
\label{bip-st}
Given a connected bipartite graph $G$ with bipartition $(P,Q)$, there is an algorithm running in $\mathcal{O}^*(2^{|Q|})$ time that computes a minimum set $X$ of vertices of $G$ such that $Q \subseteq X$ and $G[X]$ is connected. 
\end{lemma}

Our lower bound results are based on the hardness results known for the \textsc{Set Cover} problem. Given a family $\mathcal{F}$ of subsets of a universe $U$ and a positive integer $\ell$, the \textsc{Set Cover} problem is the task of determining whether there is a subfamily $\mathcal{F}' \subseteq \mathcal{F}$ of size at most $\ell$ such that $\bigcup_{X \in \mathcal{F}'} X=U$. \textsc{Set Cover} can be solved in $\mathcal{O}^*(2^{|U|})$ time by a dynamic programming routine \cite{exactFK} and the {\em Set Cover Conjecture} states that no $\mathcal{O}^*((2-\epsilon)^{|U|})$ time algorithm exists for any $\epsilon >0$~\cite{CDLMNOPSW16}. Our lower bound results are based on the following result that is a consequence of this conjecture.

\begin{lemma}{\em \cite{CDLMNOPSW16}}
\label{cvc-lb}
\cvc\ parameterized by the solution size $\ell$ has no $O^*((2-\epsilon)^{\ell})$ time algorithm under the Set Cover Conjecture. 
\end{lemma}

For each of the parameterized problems considered in this paper, we assume that the modulator (whose size is the parameter) under consideration is given as part of the input. In most cases, this assumption is reasonable. For instance, there is an algorithm by Cygan and Pilipczuk~\cite{CP13} that runs in $\OO^*(1.2732^{k+o(k)})$ time and returns a split deletion set of size at most $k$ (or decides that no such set exists). In the case of cluster deletion set size as parameter, we use the  algorithm by Boral et al.~\cite{BCKP16} that runs in $\OO^*(1.9102^k)$ time and returns a cluster deletion set of size at most $k$, if one exists. Also, a degree-$1$ modulator $S$ of size at most $k$, if one exists, can be obtained in $\OO^*(1.882^k)$ time using the algorithm by Wu~\cite{Wu2015cocoon}. Finally, there is an algorithm by Marx~\cite{Marx10algo} that runs in $\OO^*(2^{\OO(k\log k)})$ time and either outputs a chordal deletion set of size $k$ or decides that no such set exists. 

\section{Preliminaries} 
\label{sec:prelims}
We refer to \cite{Diestel12book} for graph theoretic terms and notation that are not explicitly defined here. We use $[r]$ to denote the set $\{1,2,\ldots,r\}$. Given a finite set $A$, we use ${{A}\choose{k}}$ to denote the collection of subsets of $A$ containing exactly $k$ elements and we use ${{A}\choose{\leq k}}$ to denote the collection of subsets of $A$ containing at most $k$ elements.

For a graph $G$, $V(G)$ and $E(G)$ denote the set of vertices and edges respectively. Two vertices $u, v$ are said to be \emph{adjacent} if there is an edge $\{u,v\}$ in the graph. The neighbourhood of a vertex $v$, denoted by $N_G(v)$, is the set of vertices adjacent to $v$ and its degree $d_G(v)$ is $|N_G(v)|$. The subscript in the notation for neighbourhood and degree is omitted if the graph under consideration is clear. For a set $S \subseteq V(G)$, $G - S$ denotes the graph obtained by deleting $S$ from $G$ and $G[S]$ denotes the subgraph of $G$ induced by set $S$. The {\em contraction} operation of an edge $e =uv$ in $G$ results in the deletion of $u$ and $v$ and the addition of a new vertex $w$ adjacent to vertices that were adjacent to either $u$ or $v$. Any parallel edges added in the process are deleted so that the graph remains simple. This operation is extended to a subset of edges and the resultant graph is oblivious to the contraction sequence. A {\em path} is a sequence of distinct vertices where every consecutive pair of vertices are adjacent. A set of pairwise non-adjacent vertices is called as an {\em independent set} and a set of pairwise adjacent vertices is called as a {\em clique}. A {\em complete graph} is a graph whose vertex set is a clique. The number of components of a graph $G$ is denoted by $\comp(G)$ and by convention if $G$ is connected then $\comp(G)$ is $1$. Two non-adjacent vertices $u$ and $v$ are called {\em false twins} if $N(u)=N(v)$.

A {\em tree-decomposition} of a graph $G$ is a pair $(\mathbb{T},\mathcal{ X}=\{X_{t}\}_{t\in V({\mathbb T})})$ such that $\mathbb{T}$ is a tree whose every node $t$ is assigned a vertex subset $X_t$, called a bag, satisfying the following properties.
\begin{itemize}
	\item $\bigcup_{t\in V(\mathbb{T})}{X_t}=V(G)$.
	\item For every edge $\{x,y\}\in E(G)$ there is a $t\in V(\mathbb{T})$ such that  $\{x,y\}\subseteq X_{t}$.
	\item For every  vertex $v\in V(G)$ the subgraph of $\mathbb{T}$ induced by the set  $\{t\mid v\in X_{t}\}$ is connected.
\end{itemize}
\noindent The {\em width} of a tree decomposition is $\max_{t\in V(\mathbb{T})} |X_t| -1$ and the {\em treewidth} of $G$, denoted by $tw(G)$, is the minimum width over all tree decompositions of $G$.


\smallskip
\noindent 
{\bf Lossy Kernelization: }Parameterized complexity terminology and definitions not stated here can be found in \cite{CFKLMPPS15}. We now state terminology and definitions related to lossy kernelization given in \cite{LPRS16}. The key definition in this framework is the notion of a {\em parameterized optimization (maximization / minimization) problem} which is the parameterized analogue of an optimization problem in the theory of approximation algorithms. A {\em parameterized minimization problem} is a computable function $\Pi: \Sigma^* \times \mathbb{N} \times \Sigma^* \mapsto \mathbb{R} \cup \{\pm \infty\}$. The instances of $\Pi$ are pairs $(I,k) \in \Sigma^* \times \mathbb{N}$ and a solution for $(I,k)$ is a string $S \in \Sigma^*$ such that $|S| \leq |I|+k$. The {\em value} of a solution $S$ is $\Pi(I,k,S)$. The {\em optimum value} of $(I,k)$ is $\textsc{OPT}_{\Pi}(I, k)= \underset{S \in \Sigma^*,\, |S| \leq |I|+k}\min \Pi(I,k,S)$, and an {\em optimum solution} for $(I,k)$ is a solution $S$ such that $\Pi(I,k,S)=\textsc{OPT}_{\Pi}(I, k)$.  A \emph{parameterized maximization problem} is defined in a similar way. We will omit the subscript $\Pi$ in the notation for optimum value if the problem under consideration is clear from context.

Next, we define the notion of a {\em strict $\alpha$-approximate polynomial-time preprocessing algorithm} for $\Pi$. It is defined as a pair of polynomial-time algorithms, called the {\em reduction algorithm} and the {\em solution lifting algorithm}, that satisfy the following properties.
\begin{itemize}
\item Given an instance $(I,k)$ of $\Pi$, the reduction algorithm computes an instance $(I',k')$ of $\Pi$. 
\item Given the instances $(I,k)$ and $(I',k')$ of $\Pi$, and a solution $S'$ to $(I',k')$, the solution lifting algorithm computes a solution $S$ to $(I,k)$ such that $\frac{\Pi(I,k,S)}{\textsc{OPT}(I,k)} \leq  \max \{ \frac{\Pi(I',k',S')}{\textsc{OPT}(I',k')} , \alpha \}$.
\end{itemize}
\noindent A {\em reduction rule} is the execution of the reduction algorithm on an instance, and we say that it is applicable on an instance if the output instance is different from the input instance. An {\em $\alpha$-approximate kernelization (or $\alpha$-approximate kernel)} for $\Pi$ is an $\alpha$-approximate polynomial-time preprocessing algorithm such that the size of the output instance is upper bounded by a computable function $g:\mathbb{N}\to\mathbb{N}$ of $k$. A reduction rule is said to be {\em $\alpha$-safe} for $\Pi$ if there is a solution lifting algorithm, such that the rule together with this algorithm constitute a strict $\alpha$-approximate polynomial-time preprocessing algorithm for $\Pi$. A reduction rule is {\em safe} if it is $1$-safe. Observe that this definition is more strict that the definition of safeness in classical kernelization. A {\em polynomial-size approximate kernelization scheme (PSAKS)} for $\Pi$ is a family of $\alpha$-approximate polynomial kernelization algorithms for each $\alpha >1$. Note that, the size of an output instance of a PSAKS, when run on $(I,k)$ with approximation parameter $\alpha$, must be upper bounded by $f(\alpha) k^{g(\alpha)}$ for some functions $f$ and $g$ independent of $|I|$ and $k$. A PSAKS is said to be {\em time efficient} if the reduction algorithm and the solution lifting algorithm run in $f(\alpha) |I|^c$, for some computable function $f$. 
\noindent We encourage the reader to see ~\cite{LPRS16} for a more comprehensive discussion of these ideas and definitions.

\section{Connected Vertex Cover parameterized by Split Deletion Set}
In this section, we describe an \FPT\ algorithm and a lossy polynomial kernel for \cvc\ parameterized by the size of a split deletion set. Recall that a split deletion set is a set of vertices whose deletion results in a split graph. 

\subsection{An FPT algorithm}
\label{cvc-cliq-del-set}
In this subsection we describe an $\OO^*(3^k)$ time algorithm for \cvc\ parameterized by the size $k$ of a split deletion set. 

\begin{theorem}
\label{split-fpt}
Given a graph $G$, a split deletion set $S$ and a positive integer $\ell$, there is an algorithm that determines whether $G$ has a connected vertex cover of size at most $\ell$ in $\OO^*(3^{|S|})$ time. Moreover, no $O^*((2-\epsilon)^{|S|})$ time algorithm exists for any $\epsilon>0$ for this problem under the Set Cover Conjecture. Further, if $S$ is a clique deletion set, then there is an algorithm that runs in $\OO^*(2^{|S|})$ time. 
\end{theorem}
\begin{proof}
Let $H$ denote the split graph $G-S$ whose vertex set can be partitioned into a clique $C$ and an independent set $I$. Let $|S|=k$ and let $X^*$ be a connected vertex cover of $G$ of size at most $\ell$ (if one exists). We show that once we guess the intersection of $X^*$ with $S$ and $C$, the problem reduces to a version of the Steiner tree problem with at most $|S|$ terminals. As $X^*$ can afford to exclude at most one vertex from $C$, there are at most $|C|+1$ choices for $X^* \cap C$. Also, there are at most $2^{k}$ choices for $X^* \cap S$. Consider one such choice $(Y,Z)$ where $Y=X^* \cap C$ and $Z=X^* \cap S$. Let $T$ denote the set $Y \cup Z$. We will extend $T$ into a connected vertex cover of $G$. 
If $(C \setminus Y) \cup (S \setminus Z)$ is not an independent set, no such connected vertex cover exists and we skip to the next choice of $(Y,Z)$. 
Let us first consider the case when $I=\emptyset$. That is, $S$ is a clique deletion set of $G$. In this case, as we have already
made our choice in $S$ and $C$, and so if $G[T]$ is not connected, we can skip to the next choice of $(Y,Z)$. If none of the choices leads to the required solution, we declare that $G$ has no connected vertex cover of size at most $\ell$. The overall running time of the algorithm is $\mathcal{O}^*(2^k)$.

When $I\neq\emptyset$, observe that $G[T]$ has at most $|Z|+1$ components. Let $R$ denote the set $(N(C \setminus Y) \cup N(S \setminus Z)) \setminus (Y \cup Z)$. Clearly, $R \subseteq I$ and has to be included in the solution. See Figure \ref{fig:svd}. 
If there is a vertex $r \in R$ that has no neighbours in $T$, then $r$ cannot be connected to $T$ by vertices from $I$ as $R \subseteq I$. Therefore, if there is such a vertex $r$, then we skip to the next choice of $(Y,Z)$. Otherwise, by adding $R$ to $T$, $\comp(G[T])$ cannot increase in this process. The problem now reduces to finding a set $J \subseteq I \setminus R$ such that $G[T \cup R \cup J]$ is connected.  

\begin{figure}[h]
\centering
	\includegraphics[scale=0.35]{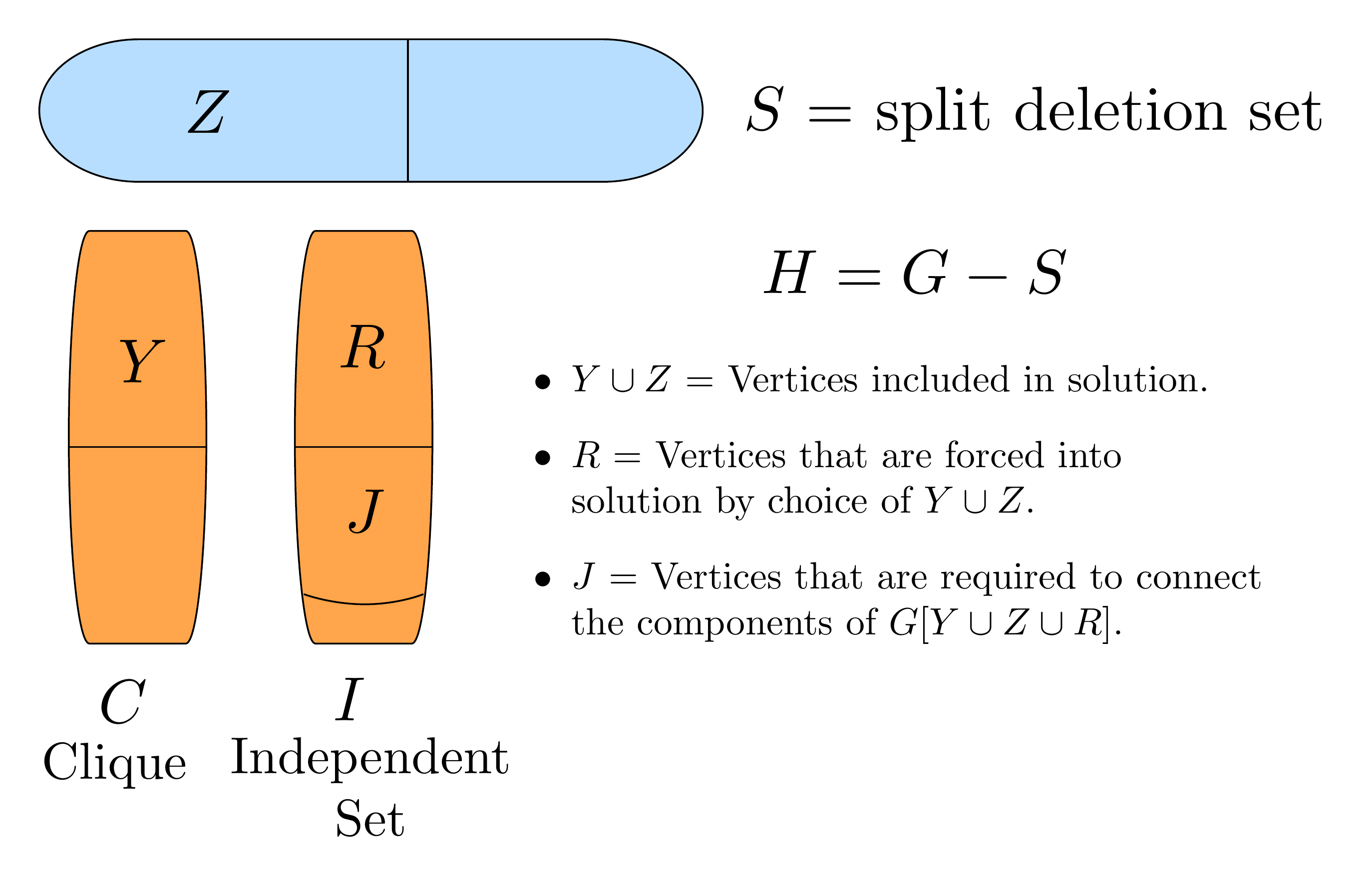}
	\caption{Connected vertex cover parameterized by split deletion set.}
	\label{fig:svd}
\end{figure}

Now consider the bipartite graph $\widehat{G}$ with bipartition $(T^*, I \setminus R)$ where each vertex of $T^*$ corresponds to a component of $G[T \cup R]$. We draw an edge between a vertex $p \in I \setminus R$ and a vertex $c$ in $T^*$ if there is an edge in $G$ between $p$ and a vertex in the component corresponding to $c$. The task now is to find a minimum Steiner tree of $\widehat{G}$ with terminal set $T^*$ which can be done in $\mathcal{O}^*(2^{|T^*|})$ time using Lemma \ref{bip-st}. Therefore, the overall running time of the algorithm is $\sum_{i=0}^{k} {k \choose i} 2^i $ where the first term in the product denotes the number of choices for $X^* \cap S$ where $|X^* \cap S|=i$ and the second term is the time taken to find a minimum Steiner tree for an instance with at most $i$ terminals. Thus, the algorithm runs in $O^*(3^k)$ time. As an edgeless graph is a split graph, we have that the size of a minimum connected vertex cover is at least the size of a minimum connected split deletion set. Therefore, the claimed lower bound follows from Lemma \ref{cvc-lb}. 
\end{proof}

In order to obtain a tighter and refined analysis of the running time of the above algorithm, we generalize the following lemma. Here $\mathcal{VC}(H)$ denotes the set of vertex covers of $H$.

\begin{lemma}{\em \cite{Cygan12swat}}
\label{no-of-vc1}
Let $H$ be a connected graph on $h$ vertices. Then, $\underset{C \in \mathcal{VC}(H)}\sum 2^{\comp(H[C])} \leq 3 \cdot 2^{h-1}$.
\end{lemma}

We next generalize Lemma \ref{no-of-vc1} to graphs that are not necessarily connected.

\begin{lemma}
\label{no-of-vc}
Let $H$ be a graph on $h$ vertices with $\comp(H)=d$. Then, $\underset{C \in \mathcal{VC}(H)}\sum 2^{\comp(H[C])} \leq 3^d 2^{h-d}$.
\end{lemma}

\begin{proof}
Let $H_1,\cdots,H_d$ denote the components of $H$ and let $h_i$ denote the number of vertices in $H_i$. Then, any vertex cover of $H$ is the union of the vertex covers of $H_i$ for $1 \leq i \leq d$. Thus, we have the following bound from Lemma \ref{no-of-vc1}. 
$$\underset{C \in \mathcal{VC}(H)}\sum 2^{\comp(H[C])} = \prod_{i=1}^{d} \underset{C \in \mathcal{VC}(H_i)}\sum 2^{\comp(H_i[C])} \leq  \prod_{i=1}^{d} 3\cdot  2^{h_i-1} = 3^d  \cdot 2^{h-d}$$
\end{proof}
 
\noindent Now, using Lemma \ref{no-of-vc} and Theorem \ref{split-fpt}, we have the following result. 

\begin{corollary}
\label{split-fpt1}
Given a graph $G$, a split deletion set $S$ with $\comp(G[S])=d$ and a positive integer $\ell$, there is an algorithm that determines whether $G$ has a connected vertex cover of size at most $\ell$ in $O^*(3^d 2^{|S|-d})$  time. 
\end{corollary}

\begin{proof}
The running time of the algorithm presented in Theorem \ref{split-fpt} is upper bounded (ignoring polynomial factors) by $\underset{Z \in \mathcal{VC}(G[S])}\sum 2^{\comp(G[Z])}$ which in turn is upper bounded by $3^d  2^{|S|-d}$ from Lemma \ref{no-of-vc}. Therefore, the claimed bound follows. 
\end{proof}

Using Lemma \ref{cvc-lb}, Corollary \ref{split-fpt1} and the fact that a vertex cover of size $k$ (if one exists) can be obtained in $\OO^*(1.2738^{\ell})$ time~\cite{CKJ01}, we have the following result. 
\begin{corollary}
Given a graph $G$, a vertex cover $S$ with $\comp(G[S])=d$ and a positive integer $\ell$, there is an algorithm that determines whether $G$ has a connected vertex cover of size at most $\ell$ in $O^*(3^d 2^{|S|-d})$ time. Further, no $O^*((2-\epsilon)^{|S|})$ time algorithm exists for this problem for any $\epsilon>0$ under the Set Cover Conjecture. 
\end{corollary}

\subsection{A Lossy Kernel}
\label{sec:lossy}
It is known that \textsc{Vertex Cover} parameterized by the size of a clique deletion set has no polynomial kernel unless \nka\ \cite{BJK14}. As \textsc{Vertex Cover} reduces to \cvc\ by just adding a universal vertex to the graph, we have the following observation.

\begin{observation}
\label{obs:cvc-clique-no-poly}
\cvc\ parameterized by the size of a clique deletion set has no polynomial kernel unless \nka.
\end{observation}

In this section, we show that \cvc\ parameterized by the size of a split deletion set admits a lossy polynomial kernel. As a consequence, we show that \cvc\ parameterized by the size of a clique deletion set admits a lossy linear kernel. The parameterized minimization problem definition of interest to us is the following. The cases specified in the definition are ordered, i.e., given $(G,S)$, an integer $k$ and a set $T \subseteq V(G)$, $CVC((G, S),k,T)$ is assigned the first value applicable. 

$CVC((G, S),k,T) = \begin{cases} -\infty &\mbox{if } |S| > k \mbox{ or } G-S \mbox{ is not a } \mbox{split graph}\\
\infty &\mbox{if } T \mbox{ is not a connected vertex cover }\\
|T| &\mbox{otherwise } \end{cases}$

\noindent Without loss of generality, we assume that $G$ has no isolated vertices as no minimum connected vertex cover contains them. Given $\alpha>1$, let $d$ be the minimum integer greater than 1 such that $\alpha \geq \frac{d-1}{d-2}$. That is, $d= \lceil \frac{2\alpha-1}{\alpha-1} \rceil$. Let $H$ denote the split graph $G-S$ whose vertex set can be partitioned into a clique $C$ and an independent set $I$. First, we bound the number of vertices in $C$ by applying Reduction Rule~\ref{rule:1}.

\begin{reduction rule}
\label{rule:1}
If $|C| \geq d$, then add a new vertex $u_{C}$ adjacent to $N(C)$ and delete $C$ to get the graph $G'$. The resulting instance is $(G',S)$.
\end{reduction rule}

Observe that $S$ continues to be a split deletion set of $G'$ as a result of applying this rule since $V(G'-S)$ can be partitioned into the clique $\{u_C\}$ and the independent set $I$.

\begin{lemma}
\label{lem:safeness-rule1}
Reduction Rule~\ref{rule:1} is $\alpha$-safe.
\end{lemma}

\begin{proof}
Consider a solution $D'$ of the reduced instance. If $u_C \in D'$, then the solution lifting algorithm returns $D=(D' \setminus \{u_C\}) \cup C$ which is a connected vertex cover of $G$ such that $\textsc{CVC}((G, S), k, D)=\textsc{CVC}((G', S), k, D')-1+|C|$. Otherwise, we know that $N(C) \subseteq D'$. Let $x$ be a vertex in $C$ such that $C \setminus \{x\}$ has a neighbour in $V(G') \setminus C$. Then, $D=D'  \cup (C \setminus \{x\})$ is a connected vertex cover of $G$ such that $\textsc{CVC}((G, S), k, D)=\textsc{CVC}((G', S), k, D')+|C|-1$. In any case, we have shown that $\textsc{CVC}((G, S), k, D)=\textsc{CVC}((G', S), k, D')+|C|-1$. Now, consider an optimum solution $D^*$ for the original instance. Clearly, $|D^* \cap C| \geq |C|-1$. Then, $(D^* \setminus C) \cup \{u_C\}$ is a connected vertex cover of $G'$. Hence, $\textsc{OPT}((G', S),k) \leq \textsc{OPT}((G, S),k)-|C|+1 +1$. Combining these bounds, we have $\frac{\textsc{CVC}((G, S), k, D)}{\textsc{OPT}((G, S),k)} \leq \max \Big\{ \frac{\textsc{CVC}((G', S),k, D')}{\textsc{OPT}((G', S),k)}, \frac{|C|-1}{|C|-2} \Big\} \leq \max\Big\{ \frac{\textsc{CVC}((G', S),k, D')}{\textsc{OPT}((G', S),k)}, \alpha\}$ as $|C| \geq d$.
\end{proof}

Observe that the application of Reduction Rule~\ref{rule:1} does not make any vertex isolated in the reduced graph. Further, when $S$ is a clique deletion set and this rule is not applicable, it follows that $G - S$ has at most $d-1$ vertices. This leads to the following result. 

\begin{theorem}
\label{thm:cvc-clique-lossy}
\cvc\ parameterized by the size $k$ of a clique deletion set has a time efficient PSAKS with at most $k + \lceil \frac{2\alpha-1}{\alpha-1} \rceil$ vertices.
\end{theorem}

\begin{proof}
Let $(G,S,\ell)$ be an instance on which Reduction Rule~\ref{rule:1} is not applicable. Then, $|S| \leq k$ and $|V(G-S)| \leq d$ where $d=\lceil \frac{2\alpha-1}{\alpha-1} \rceil$. Therefore, $|V(G)| \leq k + d$.
\end{proof}

Let us now consider the case when $S$ is not a clique deletion set. 

\begin{lemma}
\label{cvc-bound}
Let $(G,S)$ be an instance on which Reduction Rule~\ref{rule:1} is not applicable. Then, $G$ has a connected vertex cover of size at most $2k+d-1$.
\end{lemma}

\begin{proof}
If Reduction Rule~\ref{rule:1} is not applicable, then $|C| \leq d - 1$. Now, $T=S \cup C$ is a vertex cover of $G$ such that $G[T]$ has at most $|S|+1$ components. If $G[T]$ is not connected, then we can add at most $|S|$ vertices from $V(G) \setminus T$ to $T$ so that $G[T]$ becomes connected. Such a choice of vertices exists as $G$ is connected. Hence, it follows that $G$ has a connected vertex cover of size at most $|T|+|S| \leq 2k+d-1$. 
\end{proof}

From Lemma \ref{cvc-bound}, we have $\textsc{OPT}((G, S),k) \leq 2k+d-1$. Therefore, any vertex with at least $2k+d$ neighbours is present in any optimal solution. We define a partition of vertices of $G$ into the following three parts. 
\begin{itemize}
\item $B = \{u \in V(G) \mid d(u) \ge 2k+d \}$
\item $I_B = \{v \in V(G) \setminus B \mid N(v) \subseteq B \}$
\item $R = V(G) \setminus (B \cup I_B)$
\end{itemize}

We apply the following two reduction rules which are known from \cite{LPRS16} to bound $I_B$ (which is an independent set).

\begin{reduction rule}
\label{rule:cvc-high-degree}
If there exists $u \in I_B \cap V(G-S)$ such that $d_G(u) \geq d$, then delete $N_G[u]$ and add a new vertex $w$ adjacent to every vertex in $N_G(N_G(u)) \setminus \{u\}$. Further, add $2k + d$ new vertices $W$ adjacent to $w$ to get the graph $G'$. The resulting instance is $(G',(S \cup \{w\}) \setminus N(u))$.
\end{reduction rule}

Observe that $S'=(S \cup \{w\}) \setminus N(u)$ is a split deletion set of $G'$ as $V(G'-S')$ can be partitioned into the clique $C \cap V(G'-S')$ and independent set $(I \cap V(G'-S')) \cup W$. Further, as $|C| \leq d-1$ and $d_G(u) \geq d$, it follows that $N(u) \cap S \neq \emptyset$. As we add at most one vertex to $S$ and delete at least one vertex from $S$ to get $S'$, we have $|S'| \leq |S|$.

\begin{lemma}
\label{lem:safeness-cvc-high-degree}
Reduction Rule~\ref{rule:cvc-high-degree} is $\alpha$-safe.
\end{lemma}
\begin{proof}
Consider a solution $D'$ of the reduced instance. If $w \notin D'$, the solution lifting algorithm returns a connected vertex cover $D$ of $G$ of size at most $2k+d-1$ obtained from Lemma \ref{cvc-bound}. Then, we have $\textsc{CVC}((G', S'), k, D') \geq 2k+d$ since $w$ has at least $2k+d$ neighbours and $\textsc{CVC}((G, S), k, D) \leq 2k+d-1$ implying $\textsc{CVC}((G, S), k, D) \leq \textsc{CVC}((G', S'), k, D')$. Otherwise, we have $w \in D'$ and the solution lifting algorithm returns $D=(D' \setminus (W \cup \{w\})) \cup N[u]$ which is a connected vertex cover of $G$ such that $\textsc{CVC}((G, S), k, D) \leq \textsc{CVC}((G', S'), k, D')-1+|N[u]|$. Next, consider an optimum solution $D^*$ for the original instance. Clearly, $|D^*| \leq 2k+d-1$ from Lemma \ref{cvc-bound}. Thus, $B \subseteq D^*$. In particular, $N_G(u) \subseteq D^*$. Then, $(D^* \setminus N_G[u]) \cup \{w\}$ is a connected vertex cover of $G'$. Hence, $\textsc{OPT}((G', S'),k) \leq \textsc{OPT}((G, S),k)-|N_G(u)|+1$. Combining these bounds, we have $\frac{\textsc{CVC}((G, S), k, D)}{\textsc{OPT}((G, S),k)} \leq \max \Big\{ \frac{\textsc{CVC}((G', S'),k, D')}{\textsc{OPT}((G', S'),k)}, \frac{|N(u)|}{|N(u)|-1} \Big\} \leq \max \Big\{ \frac{\textsc{CVC}((G', S'),k, D')}{\textsc{OPT}((G', S'),k)}, \alpha \Big\}$ as $|N(u)| \geq d$.
\end{proof}

Recall that two non-adjacent vertices $u$ and $v$ are called false twins if $N(u)=N(v)$.

\begin{reduction rule}
\label{rule:remove-redundant}
If there exists $x \in I_B \cap V(G-S)$ such that $x$ has at least $2k+d$ false twins in $I_B \cap V(G-S)$, then delete $x$. The resulting instance is $(G-\{x\},S \setminus \{x\})$.
\end{reduction rule}

Observe that $S'=S \setminus \{x\}$ continues to be a split deletion set of $G'=G-\{x\}$ as a result of applying this rule.

\begin{lemma}
\label{lem:safeness-remove-redundant}
Reduction Rule~\ref{rule:remove-redundant} is $1$-safe.
\end{lemma}

\begin{proof}
Consider a solution $D'$ of the reduced instance. If $|D'| \geq 2k+d$, the solution lifting algorithm returns a connected vertex cover $D$ of $G$ of size at most $2k+d-1$ obtained from Lemma \ref{cvc-bound}. Then, we have $\textsc{CVC}((G', S'), k, D') \geq 2k+d$ and $\textsc{CVC}((G, S), k, D) \leq 2k+d-1$ implying $\textsc{CVC}((G, S), k, D) \leq \textsc{CVC}((G', S'), k, D')$. Otherwise, it follows that one of the false twins of $x$, say $y$, is excluded from $D'$. Thus, $N(y) \subseteq D'$ implying that $N(x) \subseteq D'$. Then, the solution lifting algorithm returns $D=D'$ which is a connected vertex cover of $G$ such that $\textsc{CVC}((G, S), k, D) = \textsc{CVC}((G', S'), k, D')$. Next, consider an optimum solution $D^*$ for the original instance. Clearly, $|D^*| \leq 2k+d-1$ from Lemma \ref{cvc-bound}. Thus, either $x$ and one of its false twins $y$ is excluded from $D^*$ or two of the false twins of $x$, say $y$ and $z$ are excluded from $D^*$. In any case, we have another optimal connected vertex cover $D^{**}$ of $G'$ that excludes $x$. Hence, $\textsc{OPT}((G', S'),k) \leq \textsc{OPT}((G, S),k)$. Combining these bounds, we have $\frac{\textsc{CVC}((G, S), k, D)}{\textsc{OPT}((G, S),k)} \leq \frac{\textsc{CVC}((G', S'), k, D')}{\textsc{OPT}((G', S'),k)}$.
\end{proof}

Now, we have the following bound.

\begin{lemma} \label{lemma:cvc-approx-kernel} 
Suppose $S$ is a split deletion set of size $k$ of $G$ and none of the Reduction Rules \ref{rule:1}, \ref{rule:cvc-high-degree} and \ref{rule:remove-redundant} is applicable on the instance $(G,S)$. Then, $|V(G)|$ is $\mathcal{O}(k+(2k+d)^{d})$.
 \end{lemma}
 \begin{proof}
We will bound $B$, $I_B$ and $R$ separately in order to bound $V(G)$. We know that $G$ has a connected vertex cover $T$ of size at most $2k+d-1$. As $B$ is the set of vertices of degree at least $2k+d$, $B \subseteq T$ and so $|B| \le 2k+d-1$. Every vertex in $R$ has degree at most $2k+d-1$. Therefore, as $T \cap R$ is a vertex cover of $G[R]$, $|E(G[R])|$ is $\mathcal{O}((2k+d-1)^2)$. Also, by the definition of $I_B$, every vertex in $R$ has a neighbour in $R$ and hence there are no isolated vertices in $G[R]$. Thus, $|R|$ is $\mathcal{O}((2k+d-1)^2)$. Finally, we bound the size of $I_B$. 

As Reduction Rule \ref{rule:cvc-high-degree} is not applicable, every vertex in $I_B \cap V(G-S)$ has degree at most $d-1$. For every set $B' \subseteq B$ of size at most $d-1$, there are at most $2k+d$ vertices in $I_B  \cap V(G-S)$ which have $B'$ as their neighbourhood. Otherwise, Reduction Rule \ref{rule:remove-redundant} would have been applied. Hence, there are at most $(2k+d) \cdot \binom{2k+(d-1)}{d-1} $ vertices in $I_B  \cap V(G-S)$. Finally, as none of the reduction rules increases the size of the split deletion set $S$, we have $|I_B  \cap S| \leq k$. Therefore, $|I_B|$ is $\mathcal{O}(k+(2k+d)^{d})$.
 \end{proof}
 
\noindent This leads to a PSAKS for the problem as claimed.

\begin{theorem}
\cvc\ parameterized by the size $k$ of a split deletion set admits a time efficient PSAKS with $\mathcal{O}(k+(2k+\lceil \frac{2\alpha-1}{\alpha-1} \rceil)^{\lceil \frac{2\alpha-1}{\alpha-1} \rceil})$ vertices. 
\end{theorem}

\begin{proof}
Given $\alpha >1$, we choose $d=\lceil \frac{2\alpha-1}{\alpha-1} \rceil$ and apply the reduction rules as long as they are applicable. Then, from Lemma \ref{lemma:cvc-approx-kernel}, $|V(G)|$ is $\mathcal{O}(k+(2k+d)^{d})$.
\end{proof}

We note that this kernel is one of the few lossy kernels for a problem with respect to a structural parameter.

\section{Connected Vertex Cover parameterized by Clique Cover}


In this section, we first show that some of the ideas from the previous section can be used to give a lossy kernel for \cvc\ when parameterized by the size of the clique cover. A clique cover of a graph is a partition of its vertex set such that each part induces a clique. We assume that we are given a partition of the vertex set of the input graph into cliques. For this parameterization, the corresponding parameterized minimization problem is defined in a similar way as defined for split deletion set. 

\begin{theorem}
\label{cc-lossy}
\cvc\ parameterized by the size $k$ of a clique cover admits a time efficient PSAKS with $\mathcal{O}(k \lceil \frac{2\alpha-1}{\alpha-1} \rceil)$ vertices. 
\end{theorem}

\begin{proof}
Without loss of generality, we assume that the graph has no isolated vertices as no minimum connected vertex cover contains them. Given $\alpha >0$, let $d= \lceil \frac{2\alpha-1}{\alpha-1} \rceil$. Let $\mathcal{C}$ denote the set of $k$ cliques in the clique cover of $G$. We bound the number of vertices in $G$ by applying Reduction Rule \ref{rule:1} on each $C \in \mathcal{C}$. From Lemma \ref{lem:safeness-rule1}, each application of this rule is $\alpha$-safe. When this rule is no longer applicable, it follows that $G$ has at most $k(d-1)$ vertices. This leads to the claimed result. 
\end{proof}

Now, we modify a reduction known from~\cite{JRV14} used in the hardness of \textsc{Feedback Vertex Set} with respect to the size of a clique cover to show that \cvc\ is $W[1]$-hard under the same parameterization.
This makes this one of the few problems that are unlikely to be in \FPT\ but have a lossy kernel.

\begin{theorem}
\label{thm:cvc-w-hard}
\cvc\ parameterized by the size of a clique cover is \W$[1]$-hard.
\end{theorem}

\begin{proof}
We reduce the well known \W$[1]$-hard problem, \textsc{Independent Set} parameterized by solution size, to our problem. A {\em non-separating independent set} of a graph $G$ is an independent set $I$ such that $V(G)\setminus I$ is a connected vertex cover of $G$. Let $(G,k)$ be an instance of \textsc{Independent Set}. We construct a graph $G'$ on the vertex set $\{(v,i) \mid v \in V(G), i \in [k]\} \cup \{x, y\}$. We add an edge between $(v,i)$ and $(u,j)$ if and only if $i = j$ or $u = v$ or $v \in N_G(u)$. Further, for every $w \in V(G') \setminus \{x\}$, we add the edge $\{w,x\}$ to $G'$. Note that $G'$ has a clique cover of size $k+1$ as for all $j \in [k]$, $G'[\{(v,j) \mid v \in V(G)\}]$ is a complete graph and $\{x,y\}$ is a clique. We claim that $G$ has an independent set of size $k$ if and only if $G'$ has a non-separating independent set of size $k+1$. 

Suppose $S= \{v_1,\ldots,v_k\}$ is an independent set of size $k$ in $G$. Define the set $S' \subseteq V(G')$ as $\{(v_1,1), (v_2,2), \ldots, (v_k,k),y\}$. Consider any two vertices $(v_i,i)$ and $(v_j,j)$ in $S'$ such that $i \neq j$. As for each distinct $i, j \in [k]$, we have $\{v_i,v_j\} \notin E(G)$, it follows that there is no edge between $(v_i,i)$ and $(v_j,j)$. Also, by the construction of $G'$, $y$ is not adjacent to any vertex in $G'$ other than $x$. Therefore, $S'$ is an independent set of size $k+1$ in $G'$. As $x$ is adjacent to every vertex in $G'-S'$, it follows that the deletion of $S'$ does not disconnect $G'$. In other words, $S'$ is a non-separating independent set of $G'$.

Conversely, suppose $S'$ is a non-separating independent set of size $k+1$ in $G'$. Clearly, different vertices of $S'$ have to come from different cliques of the clique cover of $G'$. We may assume that $y \in S'$ as if $x \in S'$, then $(S' \setminus \{x\}) \cup \{y\}$ is another non-separating independent set of size $k+1$. This is due to the facts that $y$ is not adjacent to any vertex other than $x$ and $x$ is a universal vertex in $G'$. Define the set $S \subseteq V(G)$ as $\{v \mid (v,j) \in S' \setminus \{y\}\}$. Consider any two vertices $(v,i)$ and $(u,j)$ in $S'$. Then, $u \neq v$, $i \neq j$ and $v \notin N_G(u)$. Therefore, $S$ is an independent set in $G$ of size $k$.
\end{proof}

From Theorem \ref{thm:cvc-w-hard}, it follows that \cvc\ parameterized by $k+q$ is \W$[1]$-hard where $k$ is the size of a set $S$ whose deletion results in a graph for which there exists a clique cover of size $q$. However, the problem is in \XP\ as the algorithm in Theorem \ref{split-fpt} can be easily generalized to solve \cvc\ in $O^*(2^k n^q)$ time where $n$ is the number of vertices in the input graph. Further, by an easy adaptation of Theorem \ref{cc-lossy}, the problem admits a PSAKS leading to the following result. 

\begin{theorem}
\label{thm:cvc-mod}
Given a graph $G$ on $n$ vertices, a set $S \subseteq V(G)$ such that $|S| \leq k$ and $G-S$ has a clique cover $\mathcal{Q}$ of size $q$ and a positive integer $\ell$, there is an algorithm that determines whether $G$ has a connected vertex cover of size at most $\ell$ in $O^*(2^k n^q)$ time. Further, the problem admits a time efficient PSAKS with $\mathcal{O}(k+q \lceil \frac{2\alpha-1}{\alpha-1} \rceil)$ vertices. 
\end{theorem}

\begin{proof}
Without loss of generality, we assume that the graph has no isolated vertices as no minimum connected vertex cover contains them. We also assume that the set $S$ and the clique cover $\mathcal{Q}$ of $G-S$ are part of the input. First, let us describe the \XP\ algorithm. Let $X^*$ be a connected vertex cover of $G$ of size at most $\ell$. As at least $|Q|-1$ vertices from each clique $Q \in \mathcal{Q}$ are contained in $X^*$, we guess $Y=X^* \cap (\bigcup_{Q \in \mathcal{Q}} Q)$. Then, we guess $Z=X^* \cap S$ such that $((\bigcup_{Q \in \mathcal{Q}} Q) \setminus Y) \cup (S \setminus Z)$ is an independent set. The number of choices for $(Y,Z)$ is $O^*(n^q 2^k)$. If $G[Y \cup Z]$ is not connected or has more than $\ell$ vertices, we skip to the next choice of $(Y,Z)$. Otherwise, $Y \cup Z$ is the required solution. If none of the choices leads to a solution, we declare that $G$ has no connected vertex cover of size at most $\ell$. The overall running time of the algorithm is thus $\mathcal{O}^*(2^k n^q)$.

Now, we describe an $\alpha$-approximate kernel for each $\alpha>1$. Given $\alpha >0$, let $d= \lceil \frac{2\alpha-1}{\alpha-1} \rceil$. We bound the number of vertices in $G$ by applying Reduction Rule \ref{rule:1} on each $Q \in \mathcal{Q}$. From Lemma \ref{lem:safeness-rule1}, each application of this rule is $\alpha$-safe. When this rule is no longer applicable, it follows that $G$ has at most $k+q(d-1)$ vertices. This leads to the claimed result.
\end{proof}

\section{Connected Vertex Cover parameterized by Cluster Deletion Set}

In Section \ref{cvc-cliq-del-set}, we gave an $\OO^*(2^k)$ time algorithm for \cvc\ parameterized by the size $k$ of a clique deletion set. Here, we generalize this algorithm to solve \cvc\ in $\OO^*(4^k)$ time where $k$ is the size of a cluster deletion set. Observe that this is a parameter smaller than the clique deletion set size. Further, we also describe a lossy kernel with respect to this parameterization. A classical polynomial kernel is unlikely as the size of a minimum cluster deletion set is at most the size of a minimum connected vertex cover. This is due to the fact that deleting a connected vertex cover from a graph results in a cluster graph in which every component is an isolated vertex. 

\subsection{An FPT Algorithm}
Consider an instance $(G,S,\ell)$ of \cvc\ where $S$ is a cluster deletion set. Let $H$ denote the cluster graph $G-S$. Let $X^*$ be a connected vertex cover of size at most $\ell$ (if one exists) that we are looking for. First, we guess the subset $S'$ of $S$ such that $S'=S \cap X^*$. If $S \setminus S'$ is not an independent set, we skip to the next choice of $S'$. Define the set $F$ as $N(S \setminus S') \cap V(H)$. Initialize the set $X$ to be $S'$. In the latter steps of our algorithm, we will extend $X$ to a connected vertex cover of $G$. We also update $\ell$ to $\ell - |S'|$ and delete $S \setminus S'$ from $G$. We will now describe a sequence of reduction and branching rules to be applied. The rules are applied in the order stated and a rule is applied as long as it is applicable on the instance. That is, a rule is applied only when none of the preceding rules can be applied. 

Let $I$ denote the set of isolated vertices of $H$ and $\mathcal{Q}$ denote the set of vertex sets of components of $H-I$. That is, $I$ is an independent set and each element of $\mathcal{Q}$ is a clique in $H$. These sets are updated accordingly as and when the rules are applied. First, we apply the following preprocessing rule.

\begin{preprocess rule}
\label{cluster:rule0}
If there is a clique $Q \in \mathcal{Q}$ with $N(Q) \cap X=\emptyset$ or a vertex $v \in F \cap I$ such that $N(v) \cap X=\emptyset$, then skip to the next choice of $S'$.
\end{preprocess rule}

The correctness of the first part of the rule follows from the fact that at least $|Q|-1$ vertices from $Q$ has to be in any vertex cover and any set of such vertices along with $X$ cannot be extended to a connected subgraph by only adding vertices from $V(H)$. The correctness of the second part follows from the facts that $F$ is forced into the solution and $G[X \cup \{v\}]$ cannot be extended to a connected subgraph by only adding vertices from $V(H)$. Let $F$ be partitioned into sets $Y$ and $Z$ defined as follows.
\begin{itemize}
\item $Y=\{v \in F \mid N(v) \cap X \neq \emptyset\}$, the set of vertices of $F$ that have a neighbour in $X$. 
\item  $Z=F \setminus Y$, the set of vertices of $F$ that have no neighbour in $X$.
\end{itemize}

These sets are updated over the execution of the algorithm according to the current partial solution $X$. In particular, at any point of time, $Z$ is the set of vertices in $F$ that are not adjacent to any vertex in the current partial solution $X$. Our first reduction rule is as follows.

\begin{reduction rule}
\label{cluster:rule1}
Add $Y \cup Z$ to $X$, update $\ell$ to $\ell - |Z \cup Y|$ and delete $Y$ from $H$. 
\end{reduction rule}

This rule is justified by the fact that $X^*$ must contain $F$. Due to Preprocessing Rule \ref{cluster:rule0}, every vertex $v \in Z$ is in some clique $Q$ in $\mathcal{Q}$. When Reduction Rule \ref{cluster:rule1} is no longer applicable, we have $V(H) \cap F=Z$. For each $v \in Z$, let $Q_v$ denote the clique in $\mathcal{Q}$ containing $v$. The next rule is the following. 

\begin{reduction rule}
\label{cluster:rule2}
If there is a vertex $v \in Z$ with $|Q_v \cap Z|=|Q_v|-1$, then add $Q_v \setminus Z$ to $X$, update $\ell$ to $\ell -1$ and delete $Q_v$ from $H$ and $Z$. 
\end{reduction rule}

The correctness of this rule follows from the fact that there is exactly one vertex $u$ in $Q_v \cap N(X \setminus Z)$. This vertex is forced into the solution since $Q_v \setminus \{u\} \subseteq X$ and $Q_v \setminus \{u\}$ has no neighbours in $X \setminus (Q_v \setminus \{u\})$. Now, for any clique $Q \in \mathcal{Q}$, there are at least two vertices that are not in $Z$. For a vertex $v \in V(H)$, let $\CC(v) = \{A \in G[X] \mid A$ is a component of $G[X]$ and $N_G(v) \cap V(A) \neq \emptyset\}$. The next rule is the following. 

\begin{reduction rule}
\label{cluster:rule3}
If there is a clique $Q$ in $\mathcal{Q}$ that has two vertices $u,v \notin Z$ with $\CC(u) \subseteq \CC(v)$, then update $X$ to $X \cup (Q \setminus \{u\})$ and reduce $\ell$ by the number of vertices added to $X$. Delete $Q$ from $H$ and $Z$.
\end{reduction rule}

If there exists an optimal connected vertex cover $X$ that does not contain $v$, then $u \in X$ and $X' = (X \setminus \{u\}) \cup \{v\}$ is also a connected vertex cover of $G$. This justifies the correctness of the rule. The next rule is a branching rule that is applied when $H$ has a triangle or a larger clique $Q$ with vertices not in $Z$. Further, for each $v \in Q \setminus Z$, $|\CC(v)| \geq 1$ and for every $u,v \in Q \setminus Z$, $\CC(u) \not\subseteq \CC(v)$. 

\begin{branching rule}
\label{cluster:rule4}
If there is a triangle $\{u,v,w\}$ in $H$ with $u,v,w \notin Z$, then branch into adding $u,v$ or $u,w$ or $w,v$ into $X$. In each of the branches update $\ell$ to $\ell -2$.
\end{branching rule}
The branches in this rule are clearly exhaustive as any vertex cover contains at least two vertices from a triangle. The reason for handling such special triangles will be apparent during the application of subsequent rules. When none of the rules described so far is applicable, every clique $Q \in \mathcal{Q}$ has exactly 2 vertices not in $Z$. In particular, if $Q \cap Z=\emptyset$, then $|Q| = 2$. 

\begin{reduction rule}
\label{cluster:rule5}
If there is an edge $\{u,v\}$ with $u,v \notin Z$ in $H$ and $|\CC(u)| = 1$ and $|\CC(v)| \geq 1$, then update $X$ to $X \cup (Q \setminus \{u\})$ and reduce $\ell$ by the number of vertices added to $X$ where $Q$ is the clique in $\mathcal{Q}$ that contains $u$ and $v$. Delete $Q \setminus \{u\}$ from $H$ and $Z$. Add $u$ to $I$.
\end{reduction rule}

If there exists an optimal connected vertex cover $X$ that does not contain $v$, then $u \in X$ and $X' = (X \setminus \{u\}) \cup \{v\}$ is also a vertex cover of $G$. Further, $G[X']$ is connected as the only vertex in $Q \cap X$ that is adjacent to a vertex in $X \setminus \{u\}$ is $u$ and $|\CC(u)| = 1$. Note that we crucially use the property that $Q$ has exactly two vertices adjacent to a vertex in $X \setminus \{u,v\}$ which is achieved by the application of Branching Rule \ref{cluster:rule4}. This justifies the correctness of the rule. The next rule is a branching rule that is applied until $X$ is a vertex cover (not necessarily connected) of $G$. 

\begin{branching rule}
\label{cluster:rule6}
If there is an edge $\{u,v\}$ in $H$, then in one branch we add $u$ into $X$ and in other branch we add $v$ into $X$. In both branches, update $\ell$ to $\ell - 1$ and delete the vertex added to $X$ from $H$.  
\end{branching rule}

When this rule is applied on an edge $\{u,v\}$, we have $|\CC(x)|,|\CC(y)| \geq 2$ and $\{x,y\} \cap Z=\emptyset$. The branches are exhaustive as any vertex cover contains $x$ or $y$. At this point, when none of the described rules (reduction and branching) is applicable, we have a vertex cover $X$ of $G$. That is, $V(H)$ is an independent set. However, $X$ may not necessarily induce a connected subgraph. Let $G'$ be the graph obtained from $G[V(H) \cup X]$ by contracting each component of $G[X]$ into a single vertex. $G'$ is bipartite and has a bipartition $(V(H),\widehat{X})$ where $\widehat{X}$ is the set of vertices of $G'$ corresponding to the set of components of $G[X]$. The problem is now to find a Steiner tree of $G'$ that contains $\widehat{X}$ which can be solved in $\OO^*(2^{|\widehat{X}|})$ time using Lemma \ref{bip-st}. This completes the description of the algorithm leading to the following result. 

\begin{theorem}
\label{cc-cvd-fpt}
Given a graph $G$, a cluster deletion set $S$ and a positive integer $\ell$, there is an algorithm that determines whether $G$ has a connected vertex cover of size at most $\ell$ in $O^*(4^{|S|})$ time.
\end{theorem}

\begin{proof}
Given an instance $(G,S,\ell)$, we first guess a subset $S'$ of $S$ that is contained in the solution. Then, we apply the reduction rules and branching rules in the order stated. We reiterate that each rule is applied as long as it is applicable. Also, a rule is applied only when none of the earlier rules is applicable. Let $X$ be the partial solution that we are trying to extend to a connected vertex cover. Initially $X$ is $S'$. The measure that we use to bound the running time is the number $\comp(G[X])$ of components of $G[X]$ which is at most $|S'|$.

On applying either of the branching rules, it is clear that $\comp(G[X])$ drops in all the branches by at least one. We will now show that the reductions rules do not increase the measure. When Reduction Rule \ref{cluster:rule2} is applied, $Q_v$ has a neighbour in $X$ due to Preprocessing Rule \ref{cluster:rule0}. Similarly, when Reduction Rule \ref{cluster:rule3} or \ref{cluster:rule5} is applied, $Q \setminus \{u\}$ has a neighbour in $X$. Hence, these three rules do not increase $\comp(G[X])$. Now consider Reduction Rule \ref{cluster:rule1}. Adding $Y$ to $X$ do not increase $\comp(G[X])$ while adding $Z$ to $X$ will definitely increase it. However, if $Z$ is non-empty, then in subsequent rules, some neighbour of $Z$ which is also adjacent to some vertex in (current) $X$ is added to the solution. Consider a vertex $z \in Z$ on which Reduction Rule \ref{cluster:rule2} is applicable. Then, we add a neighbour $z'$ of $z$ that is adjacent to $X$. Therefore, the measure does not increase. In fact, the measure first increases by 1 (due to addition of $z$ to $X$) and then decreases by at least 1 (due to addition of $z'$ to $X$). Suppose Reduction Rule \ref{cluster:rule2} is not applicable on a vertex $z \in Z$. Then, $Q_z$ has at least 2 vertices not in $Z$. If Reduction Rule \ref{cluster:rule3} is applicable on $Q_z$, then adding $z$ into $X$ does not increase the measure. Let us consider the case when Reduction Rule \ref{cluster:rule3} is not applicable on $Q_z$. Then, either Branching Rule \ref{cluster:rule4} or Reduction Rule \ref{cluster:rule5} or Branching Rule \ref{cluster:rule6} is applicable to vertices of $Q_z$. In any case, as we add a neighbour of $z$ that is adjacent to $X$, the  measure does not increase. 

When none of the reduction and branching rules is applicable, we have a set $X$ that is a vertex cover of $G$. That is, at each leaf of the recursion tree, we have a vertex cover $X$ of $G$ that needs to be connected by adding a minimum number of vertices from $V(H)$. We construct the bipartite graph $G'$ with bipartition $(V(H),\widehat{X})$ as described earlier. Here, $\widehat{X}$ is the set of vertices of $G'$ corresponding to the set of components of $G[X]$. Observe that $\comp(G[X])$ is at most $|S'|$.
Then, the time taken at each leaf of the recursion tree is upper bounded by the time taken to find a Steiner tree of $G'$ that contains $\widehat{X}$ which is $\OO^*(2^{|\widehat{X}|})$ from Lemma \ref{bip-st}. Let $|S'|=i$. If a leaf is at depth $j$ where $j \leq i$, it follows that the $\comp(G[X])$ is at most $i - j$. Then, the Steiner tree algorithm at this leaf runs in $\OO^*(2^{i-j})$ time. 

Let $T(j)$ denote the number of leaves of the search tree at depth $j$. It is easy to verify that the search tree is a ternary tree as there are at most three branches in any rule. Thus, the number of leaves at depth $j$ is at most $3^j$. Therefore, the running time (ignoring polynomial factors) of the algorithm is upper bounded by the following value. Let $k$ denote $|S|$.
$$\sum\limits_{i=0}^k {{k}\choose{i}} \sum\limits_{j=0}^i 3^j 2^{i-j} \leq \sum\limits_{i=0}^k i{{k}\choose{i}} 3^i= 3 \sum\limits_{i=1}^k k{{k-1}\choose{i-1}} 3^{i-1} = 3k\cdot  4^{k-1}$$
Thus, the running time of the algorithm is $\OO^*(4^k)$.
\end{proof}

A degree-$i$ modulator is a set of vertices whose deletion results in a graph with maximum degree at most $i$. As a degree-$1$ modulator is also a cluster deletion set, it follows that \cvc\ can be solved in $O^*(4^{|S|})$ time when given a degree-$1$ modulator $S$ as part of input. However, observe that this algorithm runs in $\OO^*(3^k)$ time as branching rule \ref{cluster:rule6} is never applicable (and so the search tree is a binary tree). Thus, we have the following result.
\begin{theorem}
Given a graph $G$, a degree-$1$ modulator $S$ and a positive integer $\ell$, there is an algorithm that determines whether $G$ has a connected vertex cover of size at most $\ell$ in $O^*(3^{|S|})$ time.
\end{theorem}

It is well known that the treewidth ($tw$) of a graph is not larger than the size of its minimum vertex cover, and hence of its minimum connected vertex cover. Therefore, a naive dynamic programming routine over a tree decomposition of width $tw$ solves \cvc\ in $\OO^*(2^{tw \log tw})$ time \cite{MOS05}. The running time bound has been improved to $\OO^*(2^{\OO(tw)})$ using algebraic techniques \cite{BCKN15}. If the maximum degree of a graph $G$ is upper bounded by 2, then every component of $G$ is an induced path or an induced cycle. Hence, the treewidth of $G$ is at most $2$. Therefore, if $G$ has a degree-$2$ modulator of size at most $k$, then the treewidth of $G$ is at most $k+2$. This shows that \cvc\ is \FPT\ when parameterized by the size of a degree-$2$ modulator. As \cvc\ is \NP-hard on graphs with maximum degree at most $4$, it is clear that \cvc\ when parameterized by the size of a degree-$4$ modulator is para-\NP-hard (\NP-hard for fixed values of parameter). This observation doesn't immediately carry over when parameterized by the size of a degree-$3$ modulator as \cvc\ is polynomial time solvable in sub-cubic graphs \cite{UKG88}. Nevertheless, \cvc\ is para-\NP-hard even when parameterized by the size of a degree-$3$ modulator. This is due to the fact that \textsc{Vertex Cover} in sub-cubic graphs is \NP-complete \cite{GJ79,AK2000} and \textsc{Vertex Cover} reduces to \cvc\ by just adding a universal vertex (which is a degree-$3$ modulator) to the graph. 

\begin{theorem}
\label{thm:cvc-deg-i-mod}
{\sc Connected Vertex Cover} parameterized by the size of a degree-$i$ modulator is \FPT\ when $i \leq 2$ and para-$NP$-hard for $i \geq 3$.
\end{theorem}

It is intriguing that though \cvc\ is polynomial-time solvable on sub-cubic graphs, it is \NP-hard on graphs that are just one vertex away from being a sub-cubic graph. Note that the related \textsc{Feedback Vertex Set} problem is also solvable in polynomial-time on sub-cubic graphs, but it is a major open problem whether it is polynomial-time solvable on graphs that have a degree-$3$ modulator of size 1. 

\subsection{A Lossy Kernel}
In this section, we give a PSAKS for \cvc\ parameterized by the size of a cluster deletion set. We formally define the parameterized minimization problem as follows.

$CVC((G, S),k,T) = \begin{cases} -\infty &\mbox{if } |S| > k \mbox{ or some component of } G-S \mbox{ is not a } \mbox{clique}\\
\infty &\mbox{if } T \mbox{ is not a connected vertex cover }\\
|T| &\mbox{otherwise } \end{cases}$

\noindent We now prove the following main result of this section.

\begin{theorem}
\label{thm:cluster-lossy}
\cvc\ parameterized by the size $k$ of a cluster deletion set admits a time efficient PSAKS with $\OO(k^2+\lceil \frac{2\alpha-1}{\alpha-1} \rceil \cdot \frac{k}{\alpha-1}+ \lceil \frac{\alpha}{\alpha-1} \rceil \cdot k^{\lceil \frac{\alpha}{\alpha-1} \rceil})$ vertices. 
\end{theorem}

\begin{proof}
Let $G$ be a connected graph and $S$ denote its cluster deletion set of size at most $k$. Given $\alpha >1$, define $\epsilon=\alpha-1$. Let $H$ denote the cluster graph $G-S$ and $F_0$ be the set of isolated vertices in $H$. Let  $F_1 = V(H) \setminus F_0$ and $t$ denote the number of components in $G[F_1]$. Let $C_1,\cdots,C_t$ denote the components of $G[F_1]$. Let $d_1 = \lceil \frac{2\alpha-1}{\alpha-1} \rceil$ and $d_2 = \lceil \frac{\alpha}{\alpha-1} \rceil$. Define the set $T$ as follows. We first initialize $T$ to $S$. Then, for each $i \in [t]$, we add a set $X_i$ of $|V(C_i)| - 1$ vertices of $C_i$ such that $N_G(X_i) \cap S \neq \emptyset$ to $T$. Such a set always exists as $G$ is connected. Now, $T$ is a vertex cover of $G$ and $|T| = |S| + \sum\limits_{i=1}^{t} (|V(C_i)| - 1)$. Further, $\comp(G[T])=\comp(G[S])$. If $G[T]$ is not connected, then we add at most $|S|-1$ vertices from $V(G) \setminus T$ to $T$ so that $G[T]$ becomes connected. Once again, such a set exists as $G$ is connected. Thus, we have $|T| < 2k + \sum\limits_{i=1}^{t} (|V(C_i)| - 1)$. We know that $OPT((G,S),k) \geq \sum\limits_{i=1}^{t} (|V(C_i)| - 1)$ as each $V(C_i)$ is a clique and any vertex cover excludes at most one vertex from any clique. Therefore, $|T| < 2k +OPT((G,S),k)$. If $t \geq \frac{2k}{\epsilon}$, then the reduction algorithm outputs a constant size instance $(G',S')$. Since $OPT((G,S),k) \geq t$, we have $k \leq \frac{\epsilon}{2}OPT((G,S),k)$ and it follows that $|T|$ is at most $(1+\epsilon)OPT((G,S),k)$. Equivalently, $CVC((G, S),k,T) \leq (1+\epsilon)OPT((G,S),k)$. The solution lifting algorithm simply returns $T$ which is obtained in polynomial time. 


On the other hand, suppose $t < \frac{2k}{\epsilon}$. For each $i \in [t]$, we apply Reduction Rule~\ref{rule:1} with $d=d_1$ to bound the number of vertices in $C_i$ by $d_1$. From Lemma \ref{lem:safeness-rule1}, we know that Reduction Rule~\ref{rule:1} is $\alpha$-safe. When this rule is no longer applicable, we have at most $d_1 \frac{2k}{\epsilon}$ vertices in $F_1$. Now, it remains to bound the number of vertices in $F_0$. Observe that any optimal connected vertex cover must contain any vertex $x \in S$ with $|N_G(x) \cap F_0| \geq 2k$. This is because, if some optimal connected vertex cover $T^*$ excludes $x$, then $|T^*| \geq 2k + \sum\limits_{i=1}^{t} (|V(C_i)| - 1)$. However, $T$ is a connected vertex cover with $|T| < 2k + \sum\limits_{i=1}^{t} (|V(C_i)| - 1)$ leading to a contradiction. Let $S_1 = \{x \in S \mid |N_G(x) \cap F_0| \geq 2k\}$. We first partition $F_0$ into $A_0= \{v \in F_0 \mid N_G(v) \subseteq S_1\}$ and  $B_0=F_0 \setminus A_0$. Then, we apply Reduction Rule~\ref{rule:cvc-high-degree} to vertices in $A_0$ with $d=d_2$ and Reduction Rule~\ref{rule:remove-redundant} to vertices in $A_0$ with $d=0$. Reduction Rule~\ref{rule:cvc-high-degree} is $\alpha$-safe and Reduction Rule~\ref{rule:remove-redundant} is $1$-safe from Lemmas \ref{lem:safeness-cvc-high-degree} and \ref{lem:safeness-remove-redundant}.

Suppose none of the described rules is applicable on the instance $(G,S)$. We will show that $|V(G)| = \OO(k^2+ d_2 \cdot k^{d_2}+d_1 \cdot k)$. The set $V(G)$ is partitioned into sets $S$, $F_1$, $A_0$ and $B_0$. As $G$ is connected, $G$ has no isolated vertex. As Reduction Rule~\ref{rule:1} is not applicable, each component in $G[F_1]$ has at most $d_1$ vertices. Thus, $|V(F_1)| \leq \frac{2k}{\epsilon}\cdot d_1$. For any vertex $v \in F_0$, there is a vertex $x \in S$ such that $\{x,v\} \in E(G)$. Since the reduction rules described are not applicable, every vertex $x \in A_0$ has at most $d_2-1$ neighbors in $S_1$. Further, for every $X \in {{S_1}\choose{\leq d_2-1}}$, there are at most $2k$ vertices from $A_0$ such that each of those $2k$ vertices has $X$ as its neighborhood in $G$. Therefore, $|A_0| \leq 2k\cdot \sum\limits_{i=1}^{d_2 - 1}{{k}\choose{i}} \leq 2(d_2-1)k^{d_2}$ as $|S_1| \leq k$. As $G$ is connected, each vertex in $B_0$ has a neighbour in $S \setminus S_1$. Since, the degree of any vertex in $S \setminus S_1$ is at most $2k$, it follows that $|B_0| \leq 2k^2$. Thus, $|V(G)|=|S|+|F_1|+|A_0|+|B_0|$ is at most $k+\frac{2k}{\epsilon}\cdot d_1+2(d_2-1)k^{d_2}+2k^2$ which is $\OO(k^2+\lceil \frac{2\alpha-1}{\alpha-1} \rceil \cdot \frac{k}{\alpha-1}+ \lceil \frac{\alpha}{\alpha-1} \rceil \cdot k^{\lceil \frac{\alpha}{\alpha-1} \rceil})$.
\end{proof}
As the deletion of a degree-$1$ modulator results in a graph in which every component has at most 2 vertices, we have the following result as a consequence of Theorem \ref{thm:cluster-lossy}.
\begin{corollary}
\cvc\ parameterized by the size $k$ of a degree-$1$ modulator admits a time efficient PSAKS with $\OO(k^2+\frac{k}{\alpha-1}+ \lceil \frac{\alpha}{\alpha-1} \rceil \cdot k^{\lceil \frac{\alpha}{\alpha-1} \rceil})$ vertices.
\end{corollary}

\section{Connected Vertex Cover parameterized by Chordal Deletion Set}
In this section, we show that \cvc\ parameterized by the size of a chordal deletion set is \FPT. It is well known that the treewidth ($tw$) of a graph is at most one more than the size of a minimum feedback vertex set (a set of vertices whose removal results in a forest). This immediately implies that \cvc\ is \FPT\ when parameterized by the size of a feedback vertex set. Recall that a chordal graph is a graph in which every induced cycle is a triangle. As forests, split graphs and cluster graphs are chordal, the minimum chordal deletion set size is at most a minimum feedback vertex set, a minimum split deletion set size and a minimum cluster deletion set. 

\begin{theorem}
Given a graph $G$, a chordal deletion set $S$ and a positive integer $\ell$, there is an algorithm that determines whether $G$ has a connected vertex cover of size at most $\ell$ in $O^*(2^{|S| \log |S|})$ time.
\end{theorem}
\begin{proof}
Let $G$ be a connected graph on $n$ vertices and $S$ denote its chordal deletion set. As $G-S$ is a chordal graph, a tree decomposition of $G-S$ of optimum width in which every bag is a clique can be obtained in linear time \cite{golumbic}. From this tree decomposition of $G-S$, a tree decomposition $\mathcal{T}=(T,\{X_t\}_{t \in V(T)})$ of $G$ can be obtained by adding $S$ to every bag. We will show that the algorithm for \cvc\ on bounded treewidth graphs \cite{MOS05} runs in $\OO^*(2^{\OO(|S| \log |S|)})$ time.  The algorithm employs a dynamic programming routine over the tree decomposition $T$. As this approach is largely standard, we only present an outline of the algorithm here.

For a node $t$ of $T$, let $T_t$ denote the subtree of $T$ rooted at $t \in V(T)$ and $V_t$ denote the union of all the bags of $T_t$. Let $G_t$ denote the subgraph of $G$ induced by $V_t$. For $t \in V(T)$, $X \subseteq X_t$ and a partition $\mathcal{P}=\{P_1,\cdots,P_q\}$ into at most $|X|$ parts, let $\Gamma[t,X,\mathcal{P}]$ denote a minimum vertex cover $Z$ of $G_t$ with $X_t \cap Z=X$ and $G_t[Z]$ has exactly $q$ connected components $C_1,\cdots,C_q$ where $P_i=V(C_i) \cap X_t$ for each $i \in \{1,\cdots,q\}$. Further, if $X$ is empty, $Z$ is required to be connected in $G_i$. Then, $\Gamma[t,\emptyset,\{\emptyset\}]$ is a minimum connected vertex cover of $G$. Observe that the total number of (valid) states per node is $|S|^{\OO(|S|)} n$. This is due to the fact that for each vertex $t \in V(T)$, $G[X_t]$ has a clique deletion set of size $|S|$. Further, each entry can be computed in $|S|^{\OO(|S|)} n^{\OO(1)}$ time. This leads to the claimed result.
\end{proof}
We remark that the lossy kernelization status of this problem remains open even when the chordal deletion set is a feedback vertex set.

\section{Concluding Remarks}
\label{sec:conclusion}

We have given lossy kernels and studied the parameterized complexity statuses of \cvc\ parameterized by split deletion set size, clique cover number and cluster deletion set size. Our \FPT\ running times (for \cvc\ parameterized by split deletion set and cluster deletion set) have slight gaps between upper bounds and lower bounds (based on well-known conjectures), and tightening the bounds is an interesting open problem. A more general direction is to explore the parameterized and (lossy) kernelization complexity of \cvc\ in the parameter ecology program. Figure~\ref{fig:parameter-hierarchy} gives a partial landscape of the (parameterized) complexity of \cvc\ under various structural parameters. Designing (lossy) kernels (or proving impossibility results) for those whose status is unknown is an interesting future direction.



\end{document}